%% file: extended_automata.tex
\documentclass{elsarticle}
\usepackage{amssymb}
\usepackage{hyperref}
\usepackage{graphicx}
\usepackage{url}
\usepackage{amsmath}
\usepackage{latexsym}
\usepackage{graphicx}
\usepackage{subfigure}
\usepackage{algorithm}
\usepackage{algorithmic}
\usepackage{epsf}
\usepackage{array}
\usepackage{color}
\usepackage{multirow}
\usepackage{fancyhdr}
\newcommand{\states}{\mathcal{Q}}
\newcommand{\automaton}{A}
\newcommand{\trace}{T}
\newcommand{\sequence}{H}
\newcommand{\sequenceSansPassif}{M}
\newcommand{\alphabet}{\mathcal{A}}
\newcommand{\alphabetPass}{\mathcal{T}}
\newcommand{\nodes}{\mathcal{V}}
\newcommand{\ajout}[1]{\textcolor{black}{#1}}

\newtheorem{definition}{Definition}
\newtheorem{proposition}{Proposition}
\newproof{proof}{Proof}
\newtheorem{lemma}{Lemma}
\newdimen\demilargeur
\demilargeur=\textwidth
\divide\demilargeur by 2
\begin{document}


\title{Path Computation in Multi-Layer Multi-Domain Networks: A Language Theoretic Approach\footnote{\normalsize \copyright  \ 2013. Licensed under the Creative Commons CC-BY-NC-ND 4.0 license \url{http://creativecommons.org/licenses/by-nc-nd/4.0/}}}


%
%

\author[alu]{Mohamed Lamine Lamali\corref{cor1}}
\ead{mohamed\_lamine.lamali@alcatel-lucent.com}
\author[alu]{H\'elia Pouyllau}
\ead{helia.pouyllau@alcatel-lucent.com}
\author[uvs]{Dominique Barth}
\ead{dominique.barth@prism.uvsq.fr}

\cortext[cor1]{Corresponding author}

\address[alu]{Alcatel-Lucent Bell Labs,
Route de Villejust, 91620 Nozay, France}
\address[uvs]{Lab. PR\textit{i}SM, UMR8144, Universit\'e de Versailles,
45, av. des Etas-Unis, 78035 Versailles Cedex, France}

%



%
%

\begin{abstract}
Multi-layer networks are networks in which several protocols may coexist at different layers. The Pseudo-Wire architecture provides encapsulation and decapsulation functions of protocols over Packet-Switched Networks. In a multi-domain context, computing a path to support end-to-end services requires the consideration of encapsulation and decapsulation capabilities. It appears that graph models are not expressive enough to tackle this problem. In this paper, we propose a new model of heterogeneous networks using Automata Theory. A network is modeled as a Push-Down Automaton (PDA) which is able to capture the encapsulation and decapsulation capabilities, the PDA stack corresponding to the stack of encapsulated protocols. We provide polynomial algorithms that compute the shortest path either in hops or in the number of encapsulations and decapsulations along the inter-domain path, the latter reducing manual configurations and possible loops in the path.
\end{abstract}


\maketitle 
\input{introduction}
\input{model}

\input{graph_pda}
\input{shortest_path}
\input{conclusion}
\section*{Acknowledgment}
The first author would like to thank Kaveh Ghasemloo for his help about Algorithm~\ref{calcul-ell}. This work is partially supported by the ETICS-project, funded by the European Commission. Grant agreement no.: FP7-248567 Contract Number: INFSO-ICT-248567.

\appendix
\section{List of notations}
 \label{sec:appendix}
In order to facilitate the paper reading, Table \ref{tab:symbols} summarizes the symbols used in the paper. 
\renewcommand{\arraystretch}{1}
\begin{table*}[!]
	\centering
		\begin{tabular}{|p{\demilargeur}|p{\demilargeur}|}
		\hline
		{\bf Symbols and their signification}&{\bf Symbols and their signification}\\
		\hline
		$\nodes$: The set of nodes of the network&$E$: The set of links of the network\\

		$G=(\nodes,E)$: The graph modeling the network topology&$S$: The source node\\
	
		$D$: The destination node&$\alphabet$: The alphabet (set of protocol)\\
	
		$\overline{\alphabet}$: $\{\overline{a}\ s.t.\ a\in\alphabet\}$&P(U): The set of adaptation functions of node $U$\\
	
		$\mathcal{ED}$: The set of all possible encapsulations and passive functions&$\overline{\mathcal{ED}}$: The set of all possible decapsulations\\
		
		$In(U)$: The set of protocols  that node $U$ can receive&$Out(U)$: The set of protocols that node $U$ can send\\
	
		$\alphabetPass(U)$: The set of protocols that can passively cross node $U$&$(a,a)$: Passive function\\
		
		$(a,b)$: Encapsulation of protocol $a$ in $b$&$\overline{(a,b)}$: Decapsulation of protocol $a$ from $b$\\
		
		$\trace_C$: The sequence of protocols used over the path $C$ (trace of $C$)&$\sequence_C$: Sequence of adaptation functions involved in path $C$ (transition sequence of $C$)\\
	
		$\sequenceSansPassif_C$: The well-parenthesized sequence of $C$&$\automaton_N$: Automaton corresponding to network $N$\\
	
		$\states$: The set of states of $\automaton_N$&$\Sigma$: The input alphabet of $\automaton_N$\\
	
		$\Gamma$: The stack alphabet of $\automaton_N$&$\delta$: The set of transitions of $\automaton_N$\\

		$S_\automaton$: Initial state of $\automaton_N$&$Z_0$: Initial stack symbol\\

		$D_\automaton$: Final state of $\automaton_N$&$(U_x, \langle x, \alpha,\beta\rangle, V_y)$: A transition between states $U_x$ and $V_y$, where $x$ is read, $\alpha$ is popped from the stack and $\beta$ is pushed on the stack\\

		$\states_a$: \textit{Sub\--automaton} indexed by $a$&$\automaton_N'$: Transformed PDA\\

		$L(\automaton_N)$: Language accepted by $\automaton_N$&$Nb_{push}^{\automaton_N}(w)$: Minimum number of pushes in a sequence of transitions accepting the word $w$ in $\automaton_N$\\
	
		$G_N$: Context-free grammar corresponding to network $N$&$\mathcal{N}$: The set of nonterminals of $G_N$\\
	
		$[S_G]$: Initial nonterminal (axiom) of $G_N$&$\mathcal{P}$: The set of production rules of $G_N$\\
	
		$\ell([X])$: Length of the shortest word generated by nonterminal $[X]$& \\
		\hline

		\end{tabular}
		
	\caption{List of symbols used in the paper.}
	\label{tab:symbols}
\end{table*}

\color{black}
\bibliographystyle{plain}
\bibliography{Automata}

\end{document}

%% file: introduction.tex
\section{Introduction}
Most carrier-grade networks comprise multiple layers of technologies (e.g. Ethernet, IP, etc.). These layers are administrated by different control and/or management plane instances. The Pseudo-Wire (PWE3) architecture \cite{RFC3985} unifies control plane functions in heterogeneous networks to allow multi-layer services (e.g. Layer 2 VPN). To this end, it defines encapsulation (a protocol is encapsulated into another) and decapsulation (a protocol is unwrapped from another) functions, called \textit{adaptation functions} further in this paper. These functions allow the emulation of  services (e.g. Frame Relay, SDH, Ethernet, etc.) over Packet-Switched Networks (PSN, e.g. IP or MPLS).

Prior to a service deployment in a multi-layer network, the resources must be identified during the path computation process. The path computation must take into account the adaptation function capabilities in order to explore all resources and to ensure the end-to-end service deployment. The authors of \cite{RFC5659} defined the multi-segment Pseudo-Wire architecture for multi-domain networks. In~\cite{CHOKING}, the authors mention the problem of path determination over such an architecture, stressing the importance of having path computation solutions.

In such an architecture, the path computation should comply with protocol compatibility constraint: if a protocol is encapsulated in a node, it must be decapsulated in another node; the different encapsulation processes should be transparent to the source and target nodes. Thus, some nodes may be physically connected but, due to protocol incompatibility, no feasible path (i.e., which comply with protocol compatibility) can be found between them. This constraint leads to non trivial characteristics of a shortest path \cite{Dijkstra2009}: i) it may involve loops (involving the same link several times but with different protocols); ii) its sub-paths may not be feasible. Computing such a path is challenging and cannot be performed by classical shortest path algorithms.

Currently, the configuration of these functions is manually achieved within each network domain: when an encapsulation function is used, the corresponding decapsulation function is applied within the domain boundaries.
In large-scale carrier-grade networks or in multi-domain networks, restricting the location of the adaptation functions to the domain boundaries might lead to ignore feasible end-to-end paths leading to a signaling failure. Hence, in the path computation process, it must be possible to nest several encapsulations  (e.g. Ethernet over MPLS over SDH). A decapsulation should not be restricted to the same domain as its corresponding encapsulation. This allows the exploration of more possible paths and new resources in the path computation.

The problem we address is to compute the shortest feasible path either in the number of nodes or in the number of involved (and possibly nested) adaptation functions. The latter is motivated by two goals: $i)$  as such functions are manually configured on router interfaces, minimizing their number would simplify the signaling phase when provisioning the path; $ii)$ as our algorithms do not allow loops without adaptation functions (loops without encapsulations or decapsulations are useless and can be deleted), reducing the number of adaptation functions  leads to reducing the number of loops. Reducing the number of loops is important because  if there are several loops involving the same link, the available bandwidth on this link may be exceeded.

The authors of \cite{K09,Dijkstra2008} focused on the problem of computing a path in a multi-layer network under bandwidth constraints. In \cite{LP11}, we demonstrated that the problem under multiple Quality of Service constraints is NP-Complete. In this paper we demonstrate that the problem without QoS constraints is polynomial, and we provide algorithms to compute the shortest path. These algorithms use a new model of multi-layer networks based on Push-Down Automata (PDAs). The encapsulation and decapsulation functions are designed as pushes and pops in a PDA respectively, the PDA stack allowing to memorize the nested protocols. If the goal is to minimize the number of adaptation functions, the PDA is transformed in order to bypass sub-paths without adaptation functions. The PDA or transformed PDA is then converted into a Context-Free Grammar (CFG) using a method of \cite{Hop06}. A shortest word, either corresponding to the shortest path in nodes or in adaptation functions, is generated from this CFG.



This paper extends the work published in \cite{lamali2012}, providing the detailed proofs of the correctness of the algorithms and including new algorithms (transforming the PDA, converting the PDA to a CFG, generating the shortest word). Furthermore, the complexity analyses are refined and the total worst case complexity of the path computation is significantly reduced.

This paper is organized as follows: Section~\ref{sec:motiv} recalls the context of multi-layer multi-domain networks and the related work on path computation in such networks and presents the proposed approach; Section~\ref{sec:model} provides a model of multi-layer multi-domain networks and a formal definition of the problem; Section~\ref{sec:graph_pda} explains how a network is converted into a PDA and provides the complexity of this transformation; finally, Section~\ref{sec:shortest} gives the different methods that compute the shortest path in nodes or in encapsulations/decapsulation.

In order to ease the paper reading, a table summarizing the used notations is provided in \ref{sec:appendix}.

\section{Path computation in Pseudo-Wire networks}
\label{sec:motiv}
Some standards define the emulation of lower layer protocols  over a PSN (e.g. Ethernet over MPLS, \cite{RFC4448}, Time-Division Multiplexing (TDM) over IP, \cite{RFC5087}). For instance, one node in  the network encapsulates the layer $2$ frames in layer $3$ packets and another node unwraps them. This allow to cross a part of the network by emulating a lower layer protocol, and thus to overcome protocol incompatibilities.


The PWE3 architecture \cite{RFC3985}, as well as  the multi-layer networking description of \cite{RFC5212}, assumes an exhaustive knowledge of the network states. This assumption is not valid at the multi-domain scale. Thus, the authors of \cite{RFC5659} defines an architecture for extending the Pseudo-Wire emulation across multiple PSNs segments.
The authors of \cite{CHOKING} stress the importance of path determination in such a context and suggest it to be an off-line management task. They also suggest to use the Path Computation Element architecture (PCE) \cite{PCE}, which is adapted to the multi-domain context. It could be a control plane container for solution detailed in this paper. It would require protocol and data structure extensions in order to add encapsulation/decapsulation capabilities in the PCE data model. 



\subsection{Related work on path computation in multi-layer networks} 
The problem of path computation in heterogeneous networks raised first at the optical layer. Due to technology incompatibilities (different wavelengths, different encodings, etc.), it soon became clear that classical graph models cannot capture these incompatibilities thus forbidding classical routing algorithms to tackle this problem. The authors of  \cite{Chlamtac1996} propose a \textit{Wavelength Graph Model} instead of the classical network graph models to resolve the problem of wavelength continuity. The authors of \cite{Zhu2003} propose an \textit{Auxiliary Graph Model} to resolve the problems of traffic grooming and wavelength continuity in heterogeneous WDM mesh networks, this model is later simplified in \cite{Yao2005}. The common feature of these works is that they model each physical device as several nodes, each node being indexed by a technology. The existence of an edge depends on the existence of a physical link, but also on the technology compatibility. In \cite{Gong2008}, the authors take into account the compatibility constraints on several layers: wavelength continuity, label continuity, etc. They propose a \textit{Channel Graph Model} to resolve the multi-layer path computation problem. The proposed models and algorithms take into account protocol and technology but ignore the encapsulation/decapsulation (adaptation) capabilities. As they do not have a stack to store the encapsulated protocols, they cannot model the PWE3 architecture. 

In the PWE3 architecture, the compatibility between two technologies on a layer depends also on the encapsulated protocols on the lower layer. In \cite{Dijkstra2008}, the authors addressed the problem of computing the shortest path in the context of the ITU-T G.805 recommendations on adaptation functions. They stress the lack of solutions on path selection and the limitations of graph theory to handle this problem. The authors of \cite{Dijkstra2009} present the specificities of a multi-layer path computation taking into account the encapsulation/decapsulation (adaptation) capabilities: the shortest path can contain loops and its sub-paths may not satisfy the compatibility constraints. They provide an example of topology where classical routing algorithms cannot find the shortest path because it contains a loop. In \cite{K09}, the authors addressed the same problem in a multi-layer network qualifying it as an NP-Complete problem. The NP-Completeness comes from the problem definition as they consider that the loops in the path can exceed the available bandwidth, because if the path involves the same link several times it may overload this link. They aim to select the shortest path in nodes and provide new graph models allowing to express the adaptation capabilities. They propose a Breadth-First Search algorithm which stores all possible paths and selects the shortest one. This algorithm has an exponential time complexity.

In the problem we consider, we exclude bandwidth constraints and propose a solution for minimizing the number of encapsulations and decapsulations. Our algorithm does not allow loops without adaptation functions, the only loops that may exist involve encapsulations or decapsulations. Thus, minimizing the number of adaptation functions in the path also leads to minimizing the number of loops - and avoiding them if a loop-free feasible path with less encapsulations exists.

\subsection{Proposed approach}
In this work, we propose a new multi-domain multi-layer network model wich takes into account encapsulation and decapsulation capabilities.
To the best of our knowledge no previous work has considered this problem at the multi-domain scale. It induces to go beyond the domain boundaries allowing\textit{ multi-domain compatibility} to determine a \textit{feasible inter-domain path}: when an encapsulation for a given protocol is realized in one domain its corresponding decapsulation must be done in another.  It appears that PDAs are naturally adapted to model the encapsulation and decapsulation capabilities, as push and pop operations easily model encapsulations and decapsulations, and the PDA stack can model the stack of encapsulated protocols. By using powerful tools of Automata and Language Theory, we propose polynomial algorithms that generate the shortest sequence of protocols of a feasible path. This sequence allows to find the shortest feasible path.

Furthermore, we consider two kind of objectives: either the well-known objective of minimizing the number of hops or the objective of minimizing the number of adaptation functions. The latter is motivated by the fact that it is equivalent to minimize the number of configuration operations, which are often done manually and can be quite complex. It is also motivated by reducing the number of possible loops (as the number of adaptation functions involved in a path is correlated with the number of loops), and thus avoiding to use the same link several times and to exhaust the available bandwidth on it.

Figure \ref{fig:approach} summarizes our proposed approach. It presents the different models leading to the shortest path and the algorithms computing them. 

\begin{enumerate}[I.]
	\item Convert a multi-domain Pseudo-Wire network into a PDA,
	\begin{enumerate}[i.]
		\item If the goal is to minimize the number of adaptation functions, transform the PDA to bypass the ``passive" functions (i.e. no protocol adaptation),
		\item Else let the PDA as is,
	\end{enumerate}
	\item Derive a CFG from the PDA or the transformed PDA,
	\item Determine the ``shortest" word generated by the CFG,
	\item Identify the shortest path from the shortest word.
\end{enumerate}

\begin{figure}[htb]
\center
   \includegraphics[scale=0.37]{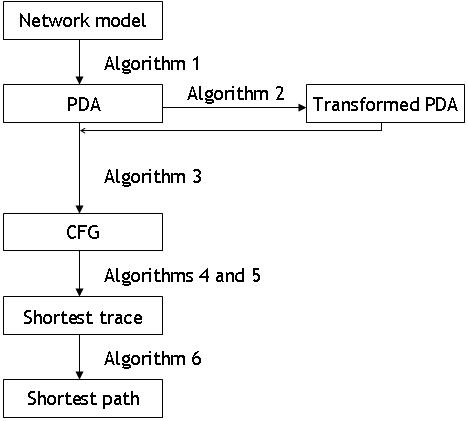}
   \caption{Proposed approach to compute the shortest feasible path.}
\label{fig:approach}
\end{figure}

Compared to the preliminary version of this work \cite{lamali2012}, we detail the proofs of correctness and refine the complexity of our algorithms. We also provide the detailed algorithm which converts the PDA to a CFG, and we propose a new method to generate the shortest word which is linear in the length of the shortest word.

\ref{sec:appendix} summarizes the notations used in this paper.

%% file: model.tex
\section{Multi-layer multi-domain network model}
\label{sec:model}

A multi-domain network having routers with encapsulation/decapsulation capabilities can be defined as a $3$-tuple: a directed graph $G=(\mathcal{V},E)$ modeling the routers of a multi-domain network, we consider a pair of vertices $(S,D)$ in $G$ corresponding to the source and the destination of the path we focus on; a finite alphabet ${\alphabet}=\{a, b, c,\dots\}$ in which each letter is a protocol; an encapsulation or a decapsulation function is a pair of different letters in the alphabet $\alphabet$:
\begin{itemize}
\item Figure~\ref{fig:encap} illustrates the encapsulation of the protocol $x$ by the node $U$ in the protocol $y$; 
\item Figure~\ref{fig:decap} illustrates that the protocol $x$ is unwrapped by the node $U$ from the protocol $y$; 
\item Figure~\ref{fig:passive} illustrates that the protocol $x$ transparently crosses the node $U$ (no encapsulation or decapsulation function is applied). Such pairs are referred as \emph{passive} further in this paper. 
\end{itemize}
We denote by $\mathcal{ED}$ and by $\overline{\mathcal {ED}}$ the set of all possible encapsulation functions and decapsulation functions respectively, (i.e., ${\mathcal{ED}} \subset {\alphabet}^2$). A subset $P(U)$ of $\mathcal{ED}\cup \overline{\mathcal{ED}}$ indicates the set of encapsulation, passive and decapsulation functions supported by vertex $U \in {\mathcal{V}}$. We define $In(U)=\{a \in {\alphabet}~s.t.\ \exists b \in {\alphabet}~s.t.~(a,b)\ \text{or}\ \overline{(b,a)} \in P(U) \}$ and $Out(U)=\{b \in {\alphabet}~s.t.~\exists a \in {\alphabet}~s.t.~(a,b) \mbox{ or } \overline{(b,a)} \in P(U)\}$. The set ${\alphabetPass}(U)=\{a\in \alphabet\ s.t.\ (a,a) \in P(U) \}$ is the set of protocols that can passively cross the node $U$. 

\begin{figure*}[htb]
\centering
\subfigure[Encapsulation of protocol $x$ in protocol~$y$, $(x,y)\in P(U)$]{\label{fig:encap}\includegraphics[scale=0.25]{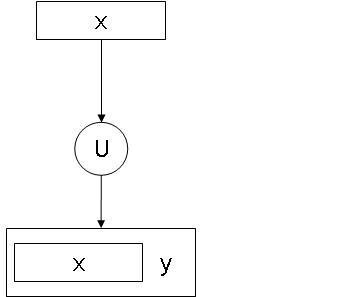}}
\hspace{10pt}
\subfigure[Decapsulation of protocol $x$ from protocol $y$, $\overline{(x,y)}\in P(U)$] {\label{fig:decap}\includegraphics[scale=0.25]{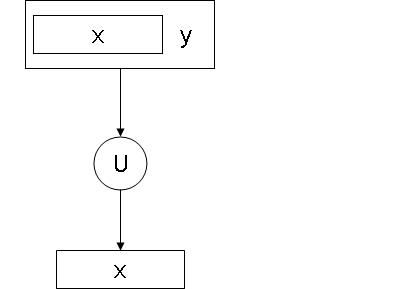}}
\hspace{10pt}
\subfigure[Passive crossing, $(x,x)\in P(U)$]{\label{fig:passive}\includegraphics[scale=0.25]{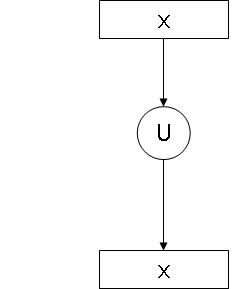}}
\caption{Different transitions when a protocol crosses a node $U$}
\label{encap-decap}
\end{figure*}

%

Considering a network $N=(G=({\mathcal{V}},E),{\alphabet},P)$, a source $S \in {\mathcal{V}}$, a destination $D \in {\mathcal{ V}}$ and a path $C=S,x^1, U_1,x^2, \ldots, U_{n-1},x^{n},D$ where each $U_i$ is a vertex in ${\mathcal{V}}$ and each $x^i \in {\alphabet} \cup \overline{\alphabet}$ (where $\overline{{\alphabet}} = \{\overline{a} : a \in {\alphabet}\}$).

\begin{itemize}
\item $\trace_C = x^1\dots x^n$ represents the sequence of protocols which is used over the path $C$. It is called the \emph{trace} of $C$. For each $x^i$:
	\begin{itemize}
		\item $x^i = a$ and $x^{i+1} = b$ or $\overline{b}$, means that $U_i$ encapsulates the protocol $a$ in $b$ ( $a,b,\overline{b} \in {\alphabet} \cup \overline{\alphabet}$) 
		\item $x^i = \overline{a}$ and $x^{i+1} = b$ or $\overline{b}$ means that $U_i$ unwraps the protocol $b$ from $a$
		\item $x^i = a$ and $x^{i+1} = a$ or $\overline{a}$ means that $U_i$ passively transports $a$
	\end{itemize}
	
\item The \textit{transition sequence} of $C$, denoted $\sequence_C$, is the sequence $\beta_1, \ldots , \beta_n$ obtained from $C$ s.t. for $i = 1..n$:
			\begin{itemize}
				\item if $x^i = a \in {\alphabet}$ and $x^{i+1} = b \in {\alphabet}$ or $x^{i+1} = \overline{b} \in \overline{\alphabet}$ then $\beta_i = (a, b)$ 
				\item if $x^i =\overline{b} \in \overline{\alphabet}$ and $x^{i+1} = a \in {\alphabet} \verb+\+ \{b\}$ or $x^{i+1} = \overline{a} \in \overline{\alphabet} \verb+\+ \{\overline{b}\}$ then $\beta_i = \overline{(a,b)}$ 
				
				 Note that the pair $(a,a),\ a\in\alphabet$ can appear in a transition sequence, as it represents a passive function. However, according the definition above, a pair $\overline{(a,a)},\ a\in\alphabet$ cannot appear, as it is forbidden to encapsulate (and thus to decapsulate) a protocol $a$ in itself.
			\end{itemize}
\item The \textit{well-parenthesized sequence} of $C$, denoted $\sequenceSansPassif_C=\beta'_1, \ldots , \beta'_m$, is obtained from $\sequence_C$ by deleting each passive transition $\beta_i$ s.t. $\beta_i=(a,a)$ and $a \in {\alphabet}$
\end{itemize}

{\noindent}{\bf Example.} Consider the path $C=S,a,U,b,V,\overline{b},W,a,D$ in the network illustrated by Fig.~ \ref{fig:network}. The transition sequence corresponding to $C$ is $\sequence_C=(a,b),(b,b),{\overline {(a,b)}}$ and its trace is $\trace_C=ab\overline{b}a$. The well-parenthesized sequence from $C$ is $\sequenceSansPassif_C=(a,b),{\overline {(a,b)}}$.

Let $\epsilon$ denotes the empty word, ``$\bullet$'' indicates the concatenation operation, and $\sequence_C$ denotes the transition sequence obtained from a path $C$ as explained above. The following definitions give a formal characterization of the feasible paths.
\begin{definition} A sequence $\sequenceSansPassif_C$ from $\sequence_C$ is \emph{valid} if and only if $\sequenceSansPassif_C \in {\mathcal{L}}$, where ${\mathcal{L}}$ is the formal language recursively defined by:
\begin{equation*} 
\mathcal{L} = \epsilon \cup \left(\bigcup_{(x,y)\in\alphabet^2}(x, y)\bullet{\mathcal{L}}\bullet\overline{(x,y)}\right)\bullet{\mathcal{L}} 
\end{equation*}
\end{definition}

\begin{definition} A path $C$ is a \emph{feasible path} in $N$ from $S$ to $D$ if:
	\begin{itemize}
		\item $S,U_1,\ldots,U_{n-1},D$ is a path in $G=(\mathcal{V},E)$,
		\item The well-parenthesized sequence $\sequenceSansPassif_C$ from $C$ is valid,
		\item For each $i \in \{1,\dots,n\}$:
				\begin{itemize}
				\item if $x^i = a$ and $x^{i+1} = b$ or $\overline{b}$ then $(a, b) \in P(U_i)$
				\item if $x^i = \overline{a}$ and $x^{i+1} = b$ or $\overline{b}$ then $\overline{(a, b)} \in P(U_i)$
				\item if $x^i = a$ and $x^{i+1} = a$ or $\overline{a}$ then $a \in {\alphabetPass}(U_i)$
				\end{itemize}
	\end{itemize}
\end{definition}
The language ${\mathcal L}$ of valid sequences is known as the \emph{generalized Dyck language} \cite{dyckwords}. It is well-known that this language is context free but not regular. Thus, push-down automata are naturally adapted to model this problem.
\\

{\noindent}{\bf Example.} The multi-domain network illustrated by Fig. \ref{fig:network} has $6$ routers and two protocols labeled $a$ and $b$. Adaptation function capabilities are indicated below each node. For example, the node $U$ can encapsulate the protocol $a$ in the protocol $b$ (function denoted by $(a,b)$) or can passively transmit the protocol $a$ (function denoted by $(a,a)$). In this multi-domain network, the only feasible path between $S$ and $D$ is $S,a,U,b,V,\bar{b},W,a,D$ and involves the encapsulation of the protocol $a$ in the protocol $b$ by the node $U$, the passive transmission of the protocol $b$ by the node $V$ and the decapsulation of the protocol $a$ from $b$ by the node $W$. (functions $(a,b)$, $(b,b)$ and $\overline{(a,b)}$ respectively).
\\

{\noindent}{\bf Problem definition.} As explained above, our goal is to find a feasible multi-domain path. Furthermore, we set as an objective function either the size of the sequence of adaptation functions or the size of the path in number of nodes. Hence, the problem we aim to solve can be defined as follows:
\begin{equation*}
\begin{split}
 &\min_C\ \   | H_C | \text{ or } | M_C |\\
 &\ \text{s.t.}\  C \text{ is a feasible path }
\end{split}
\end{equation*}


%% file: graph_pda.tex
\section{From the network model to a PDA}
\label{sec:graph_pda}
In this section, we address the conversion from a network to a PDA. Algorithm~\ref{algo1} takes as input a network $N=(G=({\mathcal{V}},E),{\alphabet},P)$ and converts it to a PDA $\automaton_N = (\states, \Sigma, \Gamma, \delta, S_\automaton, Z_0, \{D_\automaton\} )$, where $\states$ is the set of states of the PDA, $\Sigma$ the input alphabet, $\Gamma$ the stack alphabet, $\delta$ the transition relation, $S_\automaton$ the start state, $Z_0$ the initial stack symbol and $D_\automaton$ the accepting state, $\epsilon$ is the empty string. The automaton $\automaton_N$ from $N$ is obtained as follows: \

\begin{itemize}
	\item Create a state $U_x$ of the automaton for each $U \in {\mathcal V}$ and each $x \in In(U)$, except for the source node $S$ for which a single state $S_\automaton$ is created,
	\item The top of the stack is the last encapsulated protocol,
	\item If the current state is $U_x$ then the current protocol is $x$,
	\item The adaptation functions will be converted into transitions in the PDA. A transition $t\in\delta$ is denoted $(U_x, \langle x, \alpha,\beta\rangle, V_y)$ where $U_x\in\states$ is the state of the PDA before the transition, $x\in\Sigma$ is the input character, $\alpha\in\Gamma$ is the top of the stack before the transition (it is popped by the transition), $\beta\in\Gamma^*$ is the sequence of symbols pushed to the stack, and $V_y$ is the state of the PDA after the transition. Thus if $\beta=\alpha$ then $t$ is a passive transition, if $\beta=x\alpha$ then $t$ is a push of $x$, and if $\beta=\emptyset$ then $t$ is a pop,
	\item A passive transition of a protocol $x$ across a node $U$ is modeled as a transition without push or pop between the state $U_x$ and the following state $V_x$. It is denoted $(U_x, \langle x, \alpha, \alpha\rangle, V_x)$,
	\item An encapsulation of a protocol $x$ in a protocol $y$ by a node $U$ is modeled as a push of the character $x$ in the stack between the state $U_x$ and the following state $V_y$. It is denoted $(U_x, \langle x, \alpha, x\alpha\rangle, V_y)$\footnote{Note that, even if $x=\overline{a} \in \overline{\alphabet}$, the transition has the form $(U_a, \langle \overline{a}, \alpha, a\alpha\rangle, V_y)$. Characters in $\overline{\alphabet}$ are only used as input characters. Characters indexing states and pushed characters in the stack are their equivalent in $\alphabet$.} ,
	\item A decapsulation of $y$ from $x$ by a node $U$ is modeled as a pop of the protocol $y$ from the stack. It is denoted $(U_x, \langle \overline{x}, y, \emptyset\rangle, V_y)$.

\end{itemize}

\begin{algorithm}
\caption{Convert a network to a PDA}
\label{algo1}
\begin{algorithmic}
\STATE Input: a network $N = (G = ({\mathcal V},E), {\alphabet},P)$, a source $S$ and a destination $D$
\STATE Output: push-down automaton $\automaton_N = (\states, \Sigma, \Gamma, \delta, S_\automaton, Z_0, \{D_\automaton\} )$
\STATE (1) $\Sigma \gets \alphabet \cup \overline{\alphabet} $ ; $\Gamma \gets \alphabet \cup \{ Z_0\} $
\STATE (2) Create a single state $S_\automaton$ for the node $S$
\STATE (3) For each node $U\neq S$ and each protocol $x \in In(U)$, create a state $U_x$
\STATE (4) For each state $U_x$ s.t. $(S,U)\in E$ and $x\in Out(S)$
\STATE \ \ \ \ \	\ \ \ \ Create a transition $(S_\automaton, \langle\epsilon, Z_0, Z_0\rangle, U_x)$
\STATE (5) For each link $(U,V)\hspace{-0.8mm}\in\hspace{-0.8mm}E$ s.t. $U\hspace{-0.8mm}\neq\hspace{-0.8mm}S$ and for each $(x,y)\hspace{-0.8mm}\in\hspace{-0.8mm}P(U)$ and each $\alpha\hspace{-0.8mm}\in\hspace{-0.8mm}\Gamma \verb+\+ \{x\}$
\STATE \ \ \ (5.1) If $x\in {\alphabetPass}(U)\cap In(V)$ Create a transition $(U_x, \langle x, \alpha, \alpha\rangle, V_x)$\COMMENT{passive trans.}
\STATE \ \ \ (5.2) If $x\neq y$ and $y\in In(V)$ Create a transition $(U_x, \langle x, \alpha, x\alpha\rangle, V_y)$\COMMENT{encap.}
\STATE (6) For each link $(U,V) \in E$ s.t. $U\neq S$ and for each $\overline{(y,x)}\in P(U)$ 
\STATE \ \ \ (6.1) If $x\in In(V)$ Create a transition $(U_x, \langle \overline{x}, y, \emptyset\rangle, V_y)$\COMMENT{decap.} 
\STATE (7) Create a fictitious final state $D_\automaton$.
\STATE (8) For each $x\in In(D)$ and each $\alpha \in \Gamma \verb+\+ \{x\}$ Create a transition $(D_x, \langle x, Z_0, \emptyset\rangle, D_\automaton)$

\end{algorithmic}
\end{algorithm}
\noindent{\bf Complexity.} Each node $U$ from the graph generates $|In(U)|$ states, except the source node $S$. A fictitious final state is added. Thus, the number of states is at worst $2 + (|{\mathcal V}|-1)\times|{\alphabet}|$ so in $O(|{\mathcal V}|\times|{\alphabet}|)$. The worst case complexity of {algorithm}~\ref{algo1} is in $O(\max ((|{\mathcal V}|\times|{\alphabet}|), (|E|\times((|\alphabet|\times|\mathcal{ED}|) + |\overline{\mathcal{ED}}|)))))$. We assume that the network is connected, so $|E|\geq|{\mathcal V}|-1$. Since ${\mathcal{ED}}$ is a subset of ${\mathcal{A}}^2$, then $|{\mathcal{ED}}| < |{\mathcal{A}}|^2$ and $|\overline{\mathcal{ED}}| <|{\mathcal{A}}|^2$. Thus, the upper bound complexity is in $O(|\alphabet|^3\times|E|)$, which is also an upper bound for the number of transitions.

\begin{proposition}
\label{net_pda}
Considering a network $N=(G=({\mathcal V},E),{\alphabet},P)$, a source $S \in {\mathcal V}$ and a destination $D \in {\mathcal V}$, the language recognized by $\automaton_N$ is the set of traces of \emph {the feasible paths} from $S$ to $D$ in $N$.
\end{proposition}

\begin{proof}Let a feasible path $C=Sx^1U_1x^2 \ldots x^iU_i x^{i+1} \ldots U_{n-1}x^nD$ in a network $N$ and its trace $\trace_C=x^{1}\ldots x^{i}x^{i+1}\ldots x^n$
leading to a valid sequence $\sequenceSansPassif_C$. We aim to demonstrate that: $i)$ $\trace_C$ is recognized by $\automaton_N$ and $ii)$ $U_1,\ldots,U_{n-1}$ is a path in $N$ and $iii)$ There exists a feasible path corresponding to $\trace_C$ in $N$ if $\trace_C$ is recognized by the PDA $\automaton_N$. It is sufficient to show that $C=S,x^1,U_1,x^2, \ldots, U_{n-1 },x^n,D$ is a feasible path in $N$ if and only if $\trace_C$ is recognized by $\automaton_N$ (i.e the final state is reached and the stack is empty).

From $C$ we deduce the following path in $\automaton_N$ defined as a sequence of transitions. This path begins with transitions: $(S_\automaton,\langle \epsilon,Z_0,Z_0\rangle, (U_i)_{x^1})$ for each $U_i$ s.t. $x^1 \in In(U_i)$. Hence, transition $(S_\automaton,\langle \epsilon,Z_0,Z_0\rangle , (U_1)_{x^1})$ belongs to this set. Then, the path in $\automaton_N$ follows the order induced by $\trace_C$, for each element in $x^1, \ldots , x^n$, we consider one or two transitions in $\automaton_N$ as follows:
\begin{itemize}
\item if $x^i=a$ then
\begin{itemize}
\item if $x^{i+1}=a$ or $x^{i+1}=\overline{a}$ consider the transition $((U_i)_{a},\langle a,\alpha,\alpha\rangle ,(U_{i+1})_{a})$ where $\alpha \in \Gamma\verb+\+\{a\}$,
\item if $x^{i+1}=b$ or $x^{i+1}=\overline{b}$ consider the transition $((U_i)_{a},\langle a,\alpha,a\alpha\rangle ,(U_{i+1})_{b})$, where $\alpha \in \Gamma\verb+\+\{a\}$
\end{itemize}
\item else (i.e., if $x^i=\overline {a}$) then consider the transition
$( (U_i)_{a},~\langle \overline{a},b, \emptyset\rangle,~ (U_{i+1})_{b})$.
\end{itemize}
Note that, if $i=n$ then $U_{i+1}=D$.

It is clear that this sequence of transitions is a path in $\automaton_N$ recognizing the trace $\trace_C$. Since $\sequenceSansPassif_C$ is a well parenthesized word, if $x^i=\overline {a}$ then the head of the stack when reaching state $(U_i)_{x^i}$ contains $aZ_0$ and when reaching state $(U_{i+1})_{x^{i+1}}$, the head of the stack is reduced to $Z_0$. Then, when reaching state $D_\automaton$, the stack is empty.

Conversely, we can show with a similar construction that any well recognized word on a path in $\automaton_N$ induces one unique trace $T_C$ of feasible paths in $G$ from $S$ to $D$. \qed
\end{proof}

\noindent{\bf Example}. Figure~\ref{fig:pda} is an example of output of algorithm \ref{algo1}. The algorithm transforms the network illustrated by Fig. \ref{fig:network} into a PDA of Fig~\ref{fig:pda}. For instance, the link $(U,V)$ is transformed into the transitions $(U_a,\langle a,Z_0,aZ_0\rangle,V_b)$ and $(U_a,\langle a,b,ab\rangle,V_b)$ (pushes) because the node $U$ should encapsulate the protocol $a$ in the protocol $b$ using the adaptation function $(a,b)$ before transmitting to the node $V$. The link $(W,D)$ is transformed into the transitions $(W_b,\langle b,Z_0,bZ_0\rangle,D_a)$ and $(W_a,\langle b,a,ba\rangle,D_a)$ (pushes) according to the adaptation function $(a,b)$. It is also transformed into the transition $(W_b,\langle\overline{b},a,\emptyset\rangle,D_a)$ (pop) because the node $W$ can decapsulate the protocol $a$ from $b$ using the adaptation function $\overline{(a,b)}$ before transmitting to the node $D$. The link $(V,W)$ is transformed into the transitions $(V_b,\langle b,Z_0,Z_0\rangle,W_b)$ and $(V_b,\langle b,a,a\rangle,W_b)$ (passive transitions) according the capability of the node $V$ (adaptation function $(b,b)$).

%% file: shortest_path.tex
\section{The shortest feasible path}
\label{sec:shortest}
In section~\ref{sec:graph_pda}, we provided a method to build a PDA allowing to determine the feasible paths. The next step is to minimize either the number of nodes or the number of adaptation functions. The method to minimize the number of nodes uses directly the PDA as described in section~\ref{sec:min_nodes}. But to minimize the number of adaptation functions, the PDA is transformed in order to bypass the sub-paths without any adaptation function, as detailed in section~\ref{sec:pda_transformation}. Then, a CFG derived from the PDA (or the transformed PDA) generates words whose length is equivalent to the number of nodes (or to the number of adaptations). An algorithm browses the CFG to determine the shortest word. Finally, another algorithm identifies the multi-domain path corresponding to this shortest word. 
\subsection{Minimizing the number of nodes}
\label{sec:min_nodes}
The number of characters in a word accepted by the automaton $\automaton_N$ is the number of links in the corresponding feasible path (each character is a protocol used on a link). Thus the step of automaton transformation (section~\ref{sec:pda_transformation}) should be skipped. The automaton $\automaton_N$ is directly transformed into a CFG, then the shortest word is generated as described in section~\ref{sec:shortest_word}. The corresponding feasible path is computed by algorithm \ref{find-C} described in section~\ref{sec:word-path}.

\subsection{Minimizing the number of adaptation functions} 
\label{sec:pda_transformation}
To enumerate only encapsulations and decapsulations in the length of each word (and thus minimize adaptation functions by finding the shortest word accepted), a transformed automaton $A_N'$ in which all sequences involving passive transitions are bypassed must be determined. The length of the shortest word accepted by $A_N'$ is the number of adaptation functions plus a fixed constant.

Let us define $\states_a$ ($a \in \alphabet$) as $\states_a = \{V_x \in \states | x = a\}$, and let $A_N^a$ be the \emph{sub-automaton} induced by $\states_a$. By analogy with an induced subgraph, an induced sub-automaton is a multigraph with labeled edges such that the set of vertices is $\states_a$ and the set of edges is the set of transitions between elements of $\states_a$. Since there are only passive transitions between two states in $\states_a$, all paths in the sub-automaton are passive. Let define $P(U_x,V_x)$ as the shortest path length between $U_x$ and $V_x$. This length can be computed between all pairs of nodes in $\states_a$ using the Floyd\--Warshall algorithm.
Let $Succ(V_x)$ be the set of successors of $V_x$ in the original automaton $A_N$, i.e., $Succ(V_x) = \{W_y \in \states | \exists (V_x, \langle x, \alpha, \beta \rangle   , W_y) \in \delta\}$.

Algorithm \ref{algo2} takes as input $\automaton_N$ and computes the transformed automaton $\automaton_N' = (\states', \Sigma', \Gamma', \delta', S_\automaton, Z_0, \{D_A\} )$. $\automaton_N'$ is initialized with the values of $\automaton_N$. Then, the algorithm computes the sub-automaton for each character $x \in \alphabet$ (step~(1)) and the length values $P(U_x,V_x)$ for each pair of states in the sub-automaton (step (2)). Each path between a pair of states is a sequence of passive transitions. If such a path exists (step~(3.1)), the algorithm adds transitions to $\delta'$ from $U_x$ to each state in $Succ(V_x)$ (steps (3.1.2) and (3.1.3)). These added transitions are the same that those which connect $V_x$ to its successors $W_y$, but with an input character indexed by the number of passive transitions between $U_x$ and $V_x$, (i.e., $P(U_x,V_x$)) plus one (indicating that there is a transition sequence which matches with a sequence of protocols $xx\dots x$ of length $P(U_x,V_x)+1$). The indexed character is added to the input alphabet $\Sigma'$ (step (3.1.1)).
\\

\begin{algorithm}
\caption{Transform automaton $\automaton_N$}
\label{algo2}
\begin{algorithmic}
\STATE Input: push-down automaton $\automaton_N = (\states, \Sigma, \Gamma, \delta, S_A, Z_0, \{D_A\} )$
\STATE Output: transformed push-down automaton $\automaton_N' = (\states', \Sigma', \Gamma', \delta', S_A, Z_0, \{D_A\} )$
\STATE $\states' \gets \states$, $ \Sigma' \gets \Sigma$, $\Gamma' \gets \Gamma$, $\delta' \gets \delta$
\STATE For each $x \in \alphabet$
\STATE \ \ \ (1) Compute $\automaton_N^x$
\STATE \ \ \ (2) Compute the distance $P(U_x,V_x)$ between all pairs of states $U_x$ and $ V_x$ in  $\automaton_N^x$
\STATE \ \ \ (3) For each $U_x\in \states_x$ and each $V_x$ in $\states_x$
\STATE \ \ \ \ \ \ (3.1) If $P(U_x,V_x) < \infty$ \COMMENT{there is a path between $U_x$ and $V_x$}
\STATE \ \ \ \ \ \ \ \ \ \ \ (3.1.1) Add $x_{P(U_x,V_x)+1}$ and $\overline{x}_{P(U_x,V_x)+1}$ to $\Sigma$
\STATE \ \ \ \ \ \ \ \ \ \ \ (3.1.2) For each $W_y \in Succ(V_x)\backslash\{U_x\}$ and each $(V_x, \langle x, \alpha, \beta \rangle   , W_y) \in \delta$
\STATE \ \ \ \ \ \ \ \ \ \ \ \ \ \ \ \ \ \ \ \ Add the transition $(U_x, \langle x_{P(U_x,V_x)+1}, \alpha, \beta \rangle   , W_y)$ to $\delta'$
\STATE \ \ \ \ \	\ \ \ \ \ \ (3.1.3) For each $W_y \in Succ(V_x)\backslash\{U_x\}$ and each $(V_x, \langle \overline{x}, \alpha, \beta \rangle   , W_y) \in \delta$
\STATE \ \ \ \ \	\ \ \ \ \ \ \ \ \ \ \ \ \ \ \ Add the transition $(U_x, \langle \overline{x}_{P(U_x,V_x)+1}, \alpha, \beta \rangle   , W_y)$ to $\delta'$
\end{algorithmic}
\end{algorithm}

\noindent{\bf Complexity.} Computing the sub-automaton (step (1)) is done in $O(|\delta|+|\states|)$ by checking all the nodes and all the transitions. Computing the shortest path length between all pairs of states is done by the Floyd\--Warshall algorithm in $O(|\states_x|^3)$. If x=y, there are at worst $|\alphabet|-1$ transitions between $V_x$ and $W_y$ in $\states$ (transitions in the form $(U_x, \langle x, \alpha, \alpha \rangle   , W_x)$ for all possible values of $\alpha$ except $x$ \---- as the construction of the automaton forbids the encapsulation of a protocol $x$ in $x$, and thus it is impossible to have $x$ as current protocol and at the top of the stack in the same time). If $x\neq y$, there are at worst $|\alphabet|$ transitions between $V_x$ and $W_y$ in $\states$ (the pop $(V_x, \langle \overline{x}, y, \emptyset \rangle   , W_y)$ and the pushes $(V_x, \langle x, \alpha, x\alpha \rangle   , W_y)$ for all possible values of $\alpha$ except $x$). And as $|Succ(V_x)|< |\states|$, the steps (3.1.2) and (3.1.3) are bounded by $O(|\states|\times|\alphabet|)$. However, a state belongs to only one sub-automaton, the complexity of algorithm~\ref{algo2} is  therefore bounded by $O\left(\sum_{x\in\alphabet} \left[ |\delta|+|\states|+|\states_x|^3+(|\states_x|^2\times|\states|\times|\alphabet|)\right]\right)$. And $\forall x\in \alphabet, |\states_x|\leq|\alphabet|$ (see algorithm~\ref{algo1} for the construction of the set of states $\states$). Thus, the complexity of algorithm~\ref{algo2} is in $O\left(\max\left(|\alphabet|\times(|\delta|+|\states|), |\alphabet|\times|\nodes|^3,|\alphabet|^2\times|\nodes|^2\times|\states|\right)\right)$, which corresponds to $O\left(\max(|\alphabet|^4\times|E|,|\alphabet|^3\times|\nodes|^3)\right)$ in the network model. The number of transitions in $\delta'$ is bounded by $|\delta|+O\left(\sum_{x\in\alphabet}(|\states_x|^2\times|\states|\times|\alphabet|)\right)$, which corresponds to $O(|\alphabet|^3\times|\nodes|^3)$ in the network model.
\\

\noindent{\bf Example.} Algorithm \ref{algo2} transforms the PDA in Fig.~\ref{fig:pda} into the PDA in Fig.~\ref{fig:tpda}. For the protocol $b$, the algorithm computes the sub-automaton induced by the states $V_b$, $W_b$ and $K_b$. The  distance $P(V_b,W_b)=1$ is then computed. Thus, to bypass the sequence of transitions $(V_b,\langle b,a,a \rangle   ,W_b)$ $(W_b,\langle \overline{b},a,\emptyset \rangle   ,D_a)$, the transition $(V_b,\langle \overline{b_2},a,\emptyset \rangle   ,D_a)$ is added. 
\\

Let $L(\automaton_N)$ be the set of words accepted by $\automaton_N$, and let $L(\automaton_N')$ be the set of words accepted by $\automaton_N'$. Let $f : \Sigma' \rightarrow \Sigma^*$ be a function s.t.:
\begin{itemize}
\item if $x_i = a_i \in \alphabet'$ then $f(x_i) =\underbrace{aa\dots aa}_{i\ \text{occurrences}}$
\item if $x_i = \overline{a_i} \in \overline{\alphabet'}$ then $f(x_i) =\underbrace{aa\dots a\overline{a}}_{i\ \text{occurrences}}$
\end{itemize}
The domain of $f$ is extended to ${(\Sigma')}^*$:  
\begin{equation*}
\begin{split}
f : {(\Sigma')}^* &\rightarrow \Sigma^* \\
w' = x_i^1x_j^2\dots x_k^n &\rightarrow f(w') = f(x_i^1)f(x_j^2)\dots f(x_k^n)
\end{split}
\end{equation*}
 For simplicity, we consider that $x$ and $x_1$ are the same character. $f(L(\automaton_N'))$ denotes the set of words accepted by $\automaton_N'$ transformed by $f$ (i.e. $f(L(\automaton_N')) = \{ f(w')\ s.t.\ w' \in L(\automaton_N')\}$). 
 
It is clear that $f$ is not a bijection ($f(x_ix_j) = f(x_{i+j}$)). So to operate the transformation between $L(\automaton_N)$ and $L(\automaton_N')$, we define $g : \Sigma^* \rightarrow (\Sigma')^* \ s.t.:$

\noindent for each $w = \underbrace{xx\dots x}_{i\ \text{occurrences}}\underbrace{yy\dots y}_{j\ \text{occurrences}}\dots\underbrace{zz\dots z}_{k\ \text{occurrences}} \in \Sigma^*$, $g(w) = x_iy_j\dots z_k$. 

{\noindent}In other words, $w' = g(w)$ is the shortest word in $(\Sigma')^*$ s.t. $f(w') = w$.\\

The following lemmas and proposition show that the set of words accepted by the transformed automaton is \textit{`equivalent'} to the set of words accepted by the original automaton. Equivalent in the meaning that, using $f$, each word accepted by the transformed automaton can be associated to a word accepted by the original one. In addition, the transformed automaton has a new propriety: there is a linear relation between the length of an accepted word and the minimum number of pushes (and pops) involved to accept it. This propriety allows to find the word requiring the minimum number of pushes accepted by the original automaton. This is done by finding the shortest word accepted by the transformed automaton.

\begin{lemma}
\label{prop3}
$f(L(\automaton_N'))$, the set of words accepted by $\automaton_N'$ and transformed by $f$, is equal to $L(\automaton_N)$, the set of words accepted by $\automaton_N$.
\end{lemma}

\begin{proof}
We prove this by double inclusion:
\begin{enumerate}
\item $L(\automaton_N) \subseteq f(L(\automaton_N')):$ $\forall w \in L(\automaton_N), f(w) =w$ (remind that we consider that $x = x_1$). So $f(L(\automaton_N)) = L(\automaton_N)$ (1). On the other hand, the construction of $\automaton_N'$ by algorithm~\ref{algo2} does not delete any character, transition, state or set of final states from $\automaton_N$. So $L(\automaton_N) \subseteq L(\automaton_N')$. And thus $f(L(\automaton_N)) \subseteq f(L(\automaton_N'))$ (2)\\
By (1) and (2), $L(\automaton_N) \subseteq f(L(\automaton_N'))$.
\item $L(\automaton_N) \supseteq f(L(\automaton_N')):$ Let $w'=x_i^1x_j^2\dots x_l^n$ be a word in $L(\automaton_N')$ and let $t'_1t'_2\dots t'_{n+1}$ be a sequence of transitions accepting $w'$. For each transition $t'_m =(U_{x^{m-1}},\langle x_{k}^{m-1},\alpha,\beta \rangle   ,W_{x^m})$ ($1<m<n$) in this sequence, s.t. $t'_m \in \delta'\verb+\+\delta$, there is a sequence of $k-1$ passive transitions in $\automaton_N$ (because the creation of $t'_m$ in $\automaton_N'$ requires the existence of such a sequence in $\automaton_N$). This sequence begins from the state $U_{x^{m-1}}$ and is followed by a transition $(V_{x^{m-1}},\langle x^{m-1},\alpha,\beta \rangle   ,W_{x^m})$. Thus, this sequence in $\automaton_N$ matches with $f(x_{k}^{m-1})$ And for each $t'_m$ in $\automaton_N'$, either $t'_m \in \delta$ so $t'_m$ matches with $x_k^{m+1}$ (if $k=1$), or $t'_m \in \delta'\verb+\+\delta$ and there is a sequence of transitions in $\automaton_N$ which matches with $f(x_k^{m+1})$ (if $k>1$). And since, by definition, $f(w')=f(x_i^1)f(x_j^2)\dots f(x_l^n)$, $f(w') \in L(\automaton_N)$. Thus $f(L(\automaton_N')) \subseteq L(\automaton_N)$. \qed
\end{enumerate}
\end{proof}

\noindent Let $w'$ be the shortest word accepted by $\automaton_N'$, let $Nb_{push}^{\automaton_N'}(w')$ be the minimum number of pushes in a sequence of transitions accepting $w'$ in $\automaton_N'$, and let $Nb_{pop}^{\automaton_N'}(w')$ be the minimum number of pops in a sequence of transitions accepting $w'$ in $\automaton_N'$.
\begin{lemma}
\label{prop4}
The length of the shortest path (sequence of transitions) accepting $w'$ in $\automaton_N'$ is $1+Nb_{push}^{\automaton_N'}(w')+Nb_{pop}^{\automaton_N'}(w')$ and $Nb_{pop}^{\automaton_N'}(w') = Nb_{push}^{\automaton_N'}(w')+1$.
\end{lemma}

\begin{proof}
First we prove that $w'$ cannot be in the form of $\dots x_i^mx_j^{m+1}\dots$ with $x^m = x^{m+1}$. Let $U_{x^m}$ be the state in which $\automaton_N'$ is before reading the character $x_i^m$. Let $V_{x^{m+1}}$ be the state in which $\automaton_N'$ is after reading $x_i^m$ and before reading $x_j^{m+1}$. And let $W_y$ be the state in which $\automaton_N'$ is after reading $x_j^{m+1}$. If $x^m=x^{m+1}$, then, by construction, there is a transition $(U_{x^m}, \langle x_{i+j}^m,\alpha,\beta \rangle   ,W_y)$ (where $\alpha,\beta$ is a pop of $y$ or a push of $x^m$ if $x^m\neq y$). So, if $x_i^mx_j^{m+1}$ is replaced by $x_{i+j}^m$ in $w'$, the word obtained is shorter then $w'$ and is also accepted by $\automaton_N'$. Thus, the shortest word accepted by $\automaton_N'$ is $w'=x_i^1x_j^2\dots x_k^n$ where $x^m \neq x^{m+1}$  for $(1 \leq m < n)$.

By construction of $\automaton_N'$, for each character $x_i^m$ in $w'$, there is a push transition $(U_{x^m},\langle x_i^m,\alpha,x^m\alpha \rangle   ,W_y)$ if $x_i^m \in \alphabet'$ or a transition $(U_{x^m},\langle x_i^m,y,\emptyset \rangle   ,W_y)$ if $x_i^m \in \overline{\alphabet'}$. So all transitions in the shortest sequence that accepts $w'$ are pops or pushes, except the first transition from the initial state $(S_A,\langle \epsilon,Z_0,Z_0 \rangle   ,U_{x^1})$ and the final pop $(V_{x^n},\langle x_k^n,Z_0,\emptyset \rangle   ,D_A)$. The number of other pops is equal to the number of pushes (in order to have $Z_0$ at the top of the stack before the final transition).\qed
\end{proof}

\noindent Now let $w$ be a word accepted by $\automaton_N$, let $Nb_{push}^{\automaton_N}(w)$ be the minimum number of pushes in a sequence of transitions accepting $w$ in $\automaton_N$, and let $Nb_{pop}^{\automaton_N}(w)$ be the minimum number of pops in a sequence of transitions accepting $w$ in $\automaton_N$.
\begin{lemma}
\label{prop5}
For each $w \in L(\automaton_N)$:
\begin{itemize}
\item $Nb_{pop}^{\automaton_N}(w) = Nb_{pop}^{\automaton_N'}(g(w))$
\item $Nb_{push}^{\automaton_N}(w) = Nb_{push}^{\automaton_N'}(g(w))$
\end{itemize}
\end{lemma}

\begin{proof}
Let $t_1t_2\dots t_n$ be the sequence of transitions accepting $w$ in $\automaton_N$ with minimum number of pushes and pops. For each sequence $t_{i}t_{i+1}\dots t_j$ of passive transitions, followed by a push or a pop transition, there is a transition with the same push or pop and the same input character indexed by $j-i+1$. So $g(w) \in L(\automaton_N')$.

Let $t_1't_2'\dots t_{n'}'$ be the sequence of transitions with minimum pushes and pops accepting $g(w)$ in $\automaton_N'$. It is clear that $f(g(w)) = w$. So if there is a sequence $t_1''t_2''\dots t_{n'}''$ that accepts $g(w)$ with less pops and pushes than $t_1't_2'\dots t_{n'}'$, then each $t_i'' \in \delta'\verb+\+\delta$ can be replaced by a sequence of passive transitions followed by a pop or a push in $\automaton_N$. And theses sequences accept $f(g(w))$ (which is $w$) with less pops and pushes than in the sequence $t_1t_2\dots t_n$, which contradicts the fact that it minimizes the number of pops and pushes to accept $w$. \qed
\end{proof}



\begin{proposition}
\label{prop6}
The word accepted by $\automaton_N$ which minimizes the number of pushes is $f(w')$, where $w'$ is the shortest word (i.e., with minimum number of characters) accepted by $\automaton_N'$.
\end{proposition}
\begin{proof}
By lemma \ref{prop4}, the shortest word accepted by $\automaton_N'$ minimizes the number of pushes in $\automaton_N'$ (since the number of pushes grows linearly as a function of the length of the word). Let $w'$ be this word.

\noindent Suppose that there is a word $w$ accepted by $\automaton_N$ such that $Nb_{push}^{\automaton_N}(w) < Nb_{push}^{\automaton_N'}(w')$. By lemma \ref{prop5}, $Nb_{push}^{\automaton_N'}(g(w)) = Nb_{push}^{\automaton_N}(w)$, which leads to $Nb_{push}^{\automaton_N'}(g(w)) < Nb_{push}^{\automaton_N'}(w')$, while $w' = \arg \min Nb_{push}^{\automaton_N'}$. \textit{Ad absurdum}, $f(w')$ is the word which minimizes the number of pushes in $\automaton_N$.\qed
\end{proof}

\subsection{The shortest path as a shortest word} 
\label{sec:short_word}
In order to find the shortest word accepted by $\automaton_N$ (resp. $\automaton_N'$), the CFG $G_N$ such that $L(G_N)=L(A_N)$ (resp. $L(A_N')$) is computed. The shortest word in $L(G_N)$ is then generated.
\subsubsection{From the PDA to the CFG.}
\label{sec:pda-cfg}
The transformation of a PDA into a CFG is well-known. We adapted a general method described in \cite{Hop06} to transform $\automaton_N$(resp. $\automaton_N'$) into a CFG. 
The output of algorithm \ref{PDA-CFG} is a CFG $G_N = ({\mathcal N},\Sigma, [S_G],{\mathcal P})$ (resp. $({\mathcal N},\Sigma', [S_G],{\mathcal P})$) where ${\mathcal N}$ is the set of nonterminals (variables), $\Sigma$ (resp. $\Sigma'$) is the input alphabet, $[S_G]$ is the initial symbol (initial nonterminal) and ${\mathcal P}$ is the set of production rules. Except $[S_G]$, nonterminals are in the form $[UXV]$ where $U, V \in \states$ and $X \in \Gamma$ (resp. $\states'$ and $\Gamma'$). The demonstration of the correctness of this transformation is also in \cite{Hop06}.

\begin{algorithm}
\caption{Converting a PDA to a CFG}
\label{PDA-CFG}
\begin{algorithmic}
\STATE Input: PDA $\automaton_N = (\states, \Sigma, \Gamma, \delta, S_A, Z_0, \{D_A\} )$ (resp. 
Transformed PDA $\automaton_N' = (\states', \Sigma', \Gamma', \delta', S_A, Z_0, \{D_A\} )$)
\STATE Output: a CFG $G_N = ({\mathcal N},\Sigma, [S_G],{\mathcal P})$ (resp. $({\mathcal N},\Sigma', [S_G],{\mathcal P})$)
\STATE (1) Create the nonterminal $[S_G]$
\STATE (2) For each state $U_x \in \states$
\STATE \ \ \ (2.1) create a nonterminal $[S_AZ_0U_x]$ and a production $[S_G] \rightarrow [S_AZ_0U_x]$
\STATE (3) For each transition $(U_x,\langle x,\alpha,\beta \rangle,V_y)$
\STATE \ \ \ (3.1) If $\beta=\emptyset$ (pop), create a nonterminal $[U_x\alpha V_y]$ and a production $[U_x\alpha V_y] \rightarrow x$
\STATE \ \ \ (3.2) If $\beta=\alpha$ (passive transition), create for each state $W\in \states$ (resp. $\states'$)
\STATE \ \ \ \ \	\ \ \ (3.2.1) Nonterminals $[U_x\alpha W]$ and $[V_y\alpha W]$
\STATE \ \ \ \ \	\ \ \ (3.2.2) A production $[U_x\alpha W] \rightarrow x[V_y\alpha W]$
\STATE \ \ \ (3.3) If $\beta=x\alpha, \ x \in \Gamma$ (push), create for each states $(W,W')\in {\states}^2$ (resp. ${\states'}^2$)
\STATE \ \ \ \ \	\ \ \ (3.3.1) Nonterminals $[U_x\alpha W']$, $[V_y\alpha W]$ and $[W\alpha W']$
\STATE \ \ \ \ \	\ \ \ (3.3.2) A production $[U_x\alpha W'] \rightarrow x[V_yxW][W\alpha W']$
\end{algorithmic}
\end{algorithm}
\noindent{\bf Complexity.} The number of production rules in $G_N$ is bounded by  $1+|\states|+(|\delta|\times|\states|^2)$ (resp. $1+|\states'|+(|\delta'|\times|\states'|^2)$. As all the nonterminals (except $[S_G]$) are in the form $[U_x\alpha V_y]$ with $U_x,V_y \in \states$ and $\alpha \in \Gamma$, the number of non terminals is bounded by $O(|\alphabet|\times|\states|^2)$ (resp. $O(|\alphabet|\times|\states|^2)$).
The worst case complexity of algorithm \ref{PDA-CFG} is in $O(|\delta|\times|\states|^2)$ (resp. $O(|\delta'|\times|\states'|^2)$). W.r.t. the definition of the network, the upper bound is in $O(|\alphabet|^5\times|\mathcal{V}|^2\times|E|)$ (resp. $O(|\alphabet|^5\times|\mathcal{V}|^5)$).
\\

\noindent{\bf Example.} This method transforms the PDA in Fig.~\ref{fig:tpda} into a CFG. Figure~\ref{fig:cfg} is a subset of production rules of the obtained CFG. This subset allows generating the shortest trace of a feasible path in the network in Fig.~\ref{fig:network}. For instance, the transition $(V_b,\langle \overline{b_2},a,\emptyset \rangle,D_a)$ gives the production rule $[V_baD_a]\rightarrow \overline{b_2}$. The transition $(U_a,\langle a,Z_0,aZ_0\rangle,V_b)$ gives all the production rules $[U_aZ_0X']\rightarrow a[V_baX][XZ_0X']$ where $X,X' \in \states'$, including the production $[U_aZ_0D_A]\rightarrow a[V_baD_a][D_aZ_0D_A]$ .

\noindent{\bf Remark.} There are two mechanisms of acceptance defined for PDAs: an input word can be accepted either by \textit{empty stack} (i.e., if the stack is empty after reading the word) or by \textit{final state} (i.e., if a final state is reached after reading the word \-- even if the stack is not empty). 
Algorithm \ref{PDA-CFG} takes as input an automaton which accepts words by empty stack. $\automaton_N$ and $\automaton_N'$ accept words by empty stack and by final state, because the only transitions that empty the stack are those that reach the final state (the transitions in the form $(D_x,\langle x,Z_0,\emptyset \rangle   ,D_A),\ x\in In(D)$). Thus, Algorithm~\ref{PDA-CFG} performs correctly with $\automaton_N$ or $\automaton_N'$ as input.

\subsubsection{The shortest word generated by a CFG.}
\label{sec:shortest_word}
To find the shortest word generated by $G_N$, a function $\ell$ associates each nonterminal to the length of the shortest word that it generates.

More formally, $\ell :{\{{\mathcal N}\cup \Sigma \cup \{\epsilon\}\}}^* \text{or } {\{{\mathcal N}\cup \Sigma' \cup \{\epsilon\}\}}^* \rightarrow \mathbb{N}\cup\{\infty\}$ s.t.:
\begin{itemize}
\item if $w = \epsilon$ then $\ell(w) =0$,
\item if $w \in \Sigma \text{ or } \Sigma'$ then $\ell(w) = 1$,
\item if $w=\alpha_1\dots\alpha_n$ (with $\alpha_i \in \{{\mathcal N}\cup \Sigma \cup \{\epsilon\}\} \text{ or }  \{{\mathcal N}\cup \Sigma' \cup \{\epsilon\}\}$) then $\ell(w) = \sum_{i=1}^n\ell(\alpha_i)$.
\end{itemize}
Algorithm \ref{calcul-ell} computes the value of $\ell([X])$ for each $[X] \in {\mathcal N}$.
\begin{algorithm}
\caption{Compute the values $\ell([X])$ for each nonterminal $[X] \in {\mathcal N}$}
\label{calcul-ell}
\begin{algorithmic}
\STATE Input: $G_N =({\mathcal N},\Sigma, [S_G],{\mathcal P}) \text{ or } ({\mathcal N},\Sigma', [S_G],{\mathcal P})$
\STATE Output: $\ell([X])$ for each nonterminal $[X]$
\STATE (1) Initialize each $\ell([X])$ to $\infty$
\STATE (2) While there is at least one $\ell([X])$ updated at the previous iteration do
\STATE \ \ \ \ \	\ \ \ \ (2.1) For each production rule $[X] \rightarrow \alpha_1\dots\alpha_n$ in ${\mathcal P}$
\STATE \ \ \ \ \	\ \ \ \ \ \ \ \ \ (2.1.1) $\ell([X]) \gets \min\{\ell([X]),\sum_{i=1}^n\ell(\alpha_i)\}$
\end{algorithmic}
\end{algorithm}

\begin{proposition}
\label{algo-terminaison}
Algorithm \ref{calcul-ell} terminates at worst after $|{\mathcal N}|$ iterations, and each $\ell([X])\ ([X] \in {\mathcal N})$ obtained is the length (number of characters) of the shortest word produced by $[X]$.
\end{proposition}

\begin{proof}
We prove that at each iteration (except the last one), there is a least one nonterminal $[X]$ s.t. $\ell([X])$ is updated with the length of the shortest word that $[X]$ produces.Thus, the update of all $\ell([X])$ with their correct values is done at worst in $|{\mathcal N}|$ iterations. We proceed by induction on the number of iterations:\\
\\
{\bf Basis:} There is no $\epsilon$-production in $G_N$, and there is at least one production in the form $[X] \rightarrow x$ where $[X]\in {\mathcal N}$ and $x \in \Sigma$ (resp. $\Sigma'$) (see algorithm \ref{PDA-CFG} for the construction of $G_N$). For each $[X] \in {\mathcal N}\ s.t.\ \{[X]\rightarrow x\} \in {\mathcal P}, \ell([X]) = 1$. And  algorithm \ref{calcul-ell} assigns these values to each $\ell([X])$ at the first iteration.\\
\noindent{\bf Induction:} Suppose that at iteration $n$, there are at least $n'\ (n\leq n' <|{\mathcal N}|)$ nonterminals $[X]$ s.t. the algorithm has assigned the correct value to $\ell([X])$. So there are $|{\mathcal N}| - n'$ nonterminals which have not the correct $\ell$-value yet. 
\begin{itemize}
\item Either there is a nonterminal $[Y]$ s.t. $\ell([Y])$ has not already the correct value (i.e., it can already be updated), but for all productions $[Y] \rightarrow \gamma_1|\dots|\gamma_m$ (each $[Y]\rightarrow \gamma_i$ is a production)
all nonterminals in $\gamma_1,\dots,\gamma_m$ have their correct $\ell$-values. And thus, $\ell([Y])$ will be updated with the correct value at the end of the following iteration.
\item Or, for each nonterminal $[Y]$ with a wrong $\ell$-value, there is at least a nonterminal in each of its productions $[Y]\rightarrow \gamma_1|\dots|\gamma_m$ which has not its correct $\ell$-value yet.
\begin{itemize}
\item Either each $\gamma_i$ contains a nonterminal which generates no word. Thus, $[Y]$ generates no word and the value of $\ell([Y]) = \infty$ is correct.
\item Or all nonterminals in each $\gamma_i$ generate a word. However, in the derivation of the shortest word generated by any nonterminal, no nonterminal is used twice or more (otherwise a shorter word can be generated by using this nonterminal once). Thus, among all nonterminals with a wrong $\ell$-value, there is one which does not use others to generate its shortest word. So its $\ell$-value is already correct or it will be updated to the correct value at the following iteration.\qed
\end{itemize}

\end{itemize}
\end{proof}

\noindent{\bf Complexity.} As there is at worst $|{\mathcal N}|$ iteration of the \textit{while} loop, 
the complexity of algorithm \ref{calcul-ell} is in $O(|{\mathcal N}|\times|{\mathcal P}|)$ which is $O(|\alphabet|^8\times|V|^4\times|E|)$ (resp. $O(|\alphabet|^{8}\times|V|^7)$) in the network model

There are several algorithms which allow the generation of a uniform random word of some length from a CFG \cite{Hickey83,Flajolet94,Mairson94}, or, more recently, a non-uniform random generation of a word according to some targeted distribution \cite{Denise2010,Gardy2010}. 
For instance, the \emph{boustrophedonic} and the \emph{sequential} algorithms described in \cite{Flajolet94} generate a random labeled combinatorial object of some length from any decomposable structure (including CFGs). The boustrophedonic algorithm is in $O(n\log n)$ (where $n$ is the length of the object) and the sequential algorithm is in $O(n^2)$ but may have a better average complexity. Both algorithms use a precomputed table of linear size. This table can be computed in $O(n^2)$. These algorithms require an unambiguous CFG, and the CFG computed by algorithm~\ref{PDA-CFG} can be inherently ambiguous depending on the input network (as several feasible paths can use the same sequence of protocols and thus have the same trace $\trace_C$). However, the unambiguity requirement is only for the randomness of the generation. Recall that our goal is to generate the trace of the shortest feasible path. Thus, we do not take into consideration the randomness and the distribution over the set of shortest traces. In order to generate the shortest word in $L(G_N)$, the boustrophedonic algorithm can take $G_N$ and $\ell([S_G])$ as input (recall that $\ell([S_G])$ is the length of the shortest word generated by $G_N$). Thus, the generation of the shortest word $w$ (resp. $w'$) would have been in $O(|w|^2)$ (resp. $O(|w'|^2)$), where $|w|$ denotes the length (number of characters) of $w$. This complexity includes the precomputation of the table. However, this complexity hides factors depending of the size of the CFG, the construction of the precomputed table requires at least one pass over all the production rules.

Actually, as the $\ell$-values are already computed, the generation can simply be done in linear time in the length of the shortest word. The standard algorithm~\ref{short-word} takes as input the CFG and the $\ell$-values. Each nonterminal is replaced by the left part with minimal $\ell$-value among its productions.

\begin{algorithm}
\caption{Computing the shortest word generated by $G_N$}
\label{short-word}
\begin{algorithmic}
\STATE Input: $G_N = ({\cal N},\Sigma', [S_G],{\cal P})$ and all $\ell$-values
\STATE Output: $w$, the shortest word generated by $G_N$
\STATE  
\COMMENT{$[S_G] \rightarrow \gamma_1|\gamma_2|\dots|\gamma_m$ are all production rules with $[S_G]$ as left part}
\STATE (1) $w \gets \arg \min\{\ell(\gamma_1),\dots,\ell(\gamma_m)\}$
\STATE (2) While there is at least one nonterminal in $w$
	\STATE \ \ \ \ \ $(2.1)$ For each nonterminal $[X]$ in $w$ do
		\STATE  \ \ \ \ \ \ \ \ \ \ $(2.1.1)$ Replace $[X]$ in $w$ by $\arg \min \{\ell(\gamma_1'),\dots,\ell(\gamma_{m'}')\}$
		\STATE  \ \ \ \ \ \ \ \ \ \ 
		\COMMENT{$[X] \rightarrow \gamma_1'|\dots|\gamma_{m'}'$ are all production rules with $[X]$ as left part}
\end{algorithmic}
\end{algorithm}

\noindent {\bf Complexity.} Since there is no derivation of the form $[X]\rightarrow\epsilon$ in $G_N$ (see algorithm \ref{PDA-CFG}), all branches in the derivation tree end with a character from $\Sigma$ (resp. $\Sigma'$). The length of each branch is at worst $|{\cal N}|$ (as a nonterminal does not appears twice in the same branch, otherwise a shorter word could be derivated by using this nonterminal once). Thus, the number of derivations and the complexity of algorithm \ref{short-word} are bounded by $O(|{\cal N}|\times|w|)$ (resp. $O(|{\cal N}|\times|w'|)$), which corresponds to $O(|\alphabet|^3\times|\nodes|^2\times|w|)$ (resp. $O(|\alphabet|^3\times|\nodes|^2\times|w'|)$) in the network model.
\\

\noindent{\bf Example.} Algorithm \ref{calcul-ell} gives $\ell([S_G])=3$. Algorithm~\ref{short-word} computes the shortest word using the production rules in Fig.~\ref{fig:cfg}. The derivation is:
\begin{equation*}
[S_G]\overset{(1)}{\vdash} [S_AZ_0D_A]\overset{(2)}{\vdash}[U_aZ_0D_A]\overset{(3)}{\vdash}a[V_baD_a][D_aZ_0D_A] \overset{(4)}{\vdash}a\overline{b}_2[D_aZ_0D_A] \overset{(5)}{\vdash}a\overline{b}_2a
\end{equation*}
Thus, the shortest word accepted by the transformed PDA is $a\overline{b}_2a$. And the shortest trace of a feasible path is $f(a\overline{b}_2a)=ab\overline{b}a$.






\subsubsection{From the shortest word to the path.}
\label{sec:word-path}
If the goal is to minimize the number of nodes in the path, algorithm \ref{find-C} takes as input the shortest word $w$ accepted by $\automaton_N$. Otherwise, as $w'$ is the shortest word accepted by $\automaton_N'$ and generated by $G_N$, according to proposition \ref{prop6}, $f(w')$ is the word which minimizes the number of pops and pushes in $\automaton_N$. In such a case it is the trace $T_C$ of the shortest feasible path $C$ in the network $N$. It is possible that several paths match with the trace $T_C=w$ (resp. $f(w')$). In such a case, a load-balancing policy can choose a path.

Algorithm \ref{find-C} is a dynamic programming algorithm that computes $C$. It starts at the node $S$ and takes at each step all the links in $E$ which match with the current character in $T_C$. Let $T_C = x^1x^2\dots x^n$ $(x^i \in \alphabet\cup\overline{\alphabet})$. At each step $i$, the algorithm starts from each node $U$ in $Nodes[i]$ and adds to $Links[i]$ all links $(U,V)$ which match with $x^i$. It also adds each $V$ in $Nodes[i+1]$. When reaching $D$, it backtracks to $S$ and selects the links from $D$ to $S$.

\begin{algorithm}
\caption{Find the shortest path}
\label{find-C}
\begin{algorithmic}
\STATE Input: Network $N$ and $T_C$
\STATE Output: Shortest path $C$
\STATE (1) $Nodes[1] \gets S$ ; $i \gets 1$
\STATE (2) While $D$ is not reached do
\STATE \ \ \ \ \	\ \ \ \ (2.1) for each $U \in Nodes[i]$ and each $V \in {\mathcal V}\ s. t.\ (U,V) \in E$ do
\STATE \ \ \ \ \	\ \ \ \ \ \ \ \ \ (2.1.1) If $x^i\in\alphabet$, $x^i\in Out(U)$, $x^i\in In(V)$ and $(x^{i-1},x^i)\in P(U)$
\STATE \ \ \ \ \	\ \ \ \ \ \ \ \ \ \ \ \ \ \ \ \ \ \ \ \ (2.1.1.1) Add $(U,V)$ to $Links[i]$ and $V$ to $Nodes[i+1]$
\STATE \ \ \ \ \	\ \ \ \ \ \ \ \ \ (2.1.2) If $x^i \in \overline{\alphabet}$, $x^i\in Out(U)$, $x^i\in In(V)$ and $\overline{(x^{i},x^{i-1})}\in P(U)$
\STATE \ \ \ \ \	\ \ \ \ \ \ \ \ \ \ \ \ \ \ \ \ \ \ \ \ (2.1.2.1) Add $(U,V)$ to $Links[i]$ and $V$ to $Nodes[i+1]$
\STATE \ \ \ \ \	\ \ \ \ (2.2) $i++$
\STATE (3) Backtrack from $D$ to $S$ by adding each covered link in the backtracking to $C$.
\end{algorithmic}
\end{algorithm}

\noindent{\bf Complexity.} The $while$ loop stops exactly after $T_C$ steps, because it is sure that there is a feasible path of length $|T_C|$ if $T_C$ is accepted by the automaton $\automaton_N$. At each step, all links and nodes are checked in the worst case. Thus, algorithm \ref{find-C} is in $O(|T_C|\times|{\mathcal V}|\times|E|)$ in the worst case.
\\

\noindent{\bf Example.} From the shortest trace $ab\overline{b}a$, algorithm \ref{find-C} computes the only feasible path in the network on Fig.~\ref{fig:network}, which is $S, a, U, b, V, \overline{b}, W, a, D$.

%% file: conclusion.tex
\section{Conclusion}

\renewcommand{\arraystretch}{1.2}
\begin{table*}[t]
	\footnotesize
	\centering
		\begin{tabular}{|l|c|c|}
		\hline
		\multicolumn{1}{|c|}{\multirow{2}{*}{Algorithm}}&\multicolumn{2}{c|}{Upper-Bound Complexity}\\
		\cline{2-3} 
		 &  Minimizing hops & Minimizing enc./dec.\\
		 \hline
		 
		Algo. \ref{algo1}: Network to PDA& \multicolumn{2}{c|}{$O(|\alphabet|^3\times|E|)$}\\
		\hline
		Algo. \ref{algo2}: Transform PDA& / &$O(\max(|\alphabet|^4\times|E|,|\alphabet|^3\times|\nodes|^3)$\\
		\hline
		Algo. \ref{PDA-CFG}: PDA to CFG& $O(|\alphabet|^5\times|\mathcal{V}|^2\times|E|)$ & $O(|\alphabet|^5\times|{\mathcal V}|^5)$\\
		\hline
		Algo. \ref{calcul-ell}: Shortest word length& $O(|\alphabet|^8\times|{\mathcal V}|^4\times|E|)$ & $O(|\alphabet|^{8}\times|{\mathcal V}|^7)$\\
		\hline
		Algo. \ref{short-word}: Shortest word&$O(|\alphabet|^3\times|\nodes|^2\times|w|)$& $O(|\alphabet|^3\times|\nodes|^2\times|w'|)$\\
		\hline
		Algo. \ref{find-C}: Shortest path& \multicolumn{2}{c|}{$O(|T_C|\times |{\mathcal V}|\times|E|)$}\\
		\hline
		\end{tabular}
	\caption{Algorithms and their complexities}
	\label{tab:summary}
\end{table*}

The problem of path computation in a multi-layer network has been studied in the field of intra-domain path computation but is less addressed in the inter-domain field with consideration of the Pseudo-Wire architecture. There was no polynomial solution to this problem and the models used were not expressive enough to capture the encapsulation/decapsulation capabilities described in the Pseudo-Wire architecture. 

\ajout{In this paper, we provide algorithms that compute the shortest path in a multi-layer multi-domain network, minimizing the number of hops or the number of encapsulations and decapsulations. The presented algorithms involve Automata and Language Theory methods. A Push-Down Automaton models the multi-layer multi-domain network. It is then transformed in order to bypass passive transitions and converted into a Context-Free Grammar. The grammar generates the shortest protocol sequence, which allows to compute the path matching this sequence.}

The different algorithms of our methodology have polynomial upper-bound complexity as summarized by Table \ref{tab:summary}. Compared to the preliminary version of this work, the proofs of correctness of the algorithms are detailed and the complexity analysis is significantly refined (the highest algorithm complexity is in  $O(|\alphabet|^{8}\times|{\mathcal V}|^7)$ instead of  $O(|\alphabet|^{11}\times|{\mathcal V}|^7\times|E|^2)$).

In order to figure out the whole problem of end-to-end service delivery, we plan to extend our solution to support end-to-end \emph{Quality of Service} constraints and model all technology constraints on the different layers \ajout{(conversion or ``\textit{mapping}" of protocols)}. \ajout{As a future work, we also plan to investigate which part of our algorithms can be distributed (e.g., Can the domains publish their encapsulation/decapsulation capabilities without disclosing their internal topology?) and how such a solution can be established on today's architectures.}

\begin{figure}[!]
\centering
\subfigure[Network]{\label{fig:network}\includegraphics[scale=0.35]{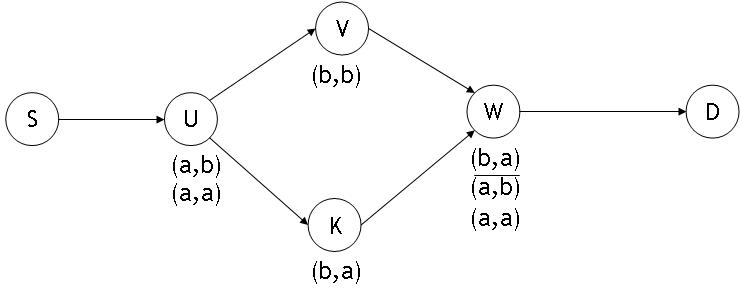}}
\subfigure[Corresponding PDA]{\label{fig:pda}\includegraphics[width=\textwidth]{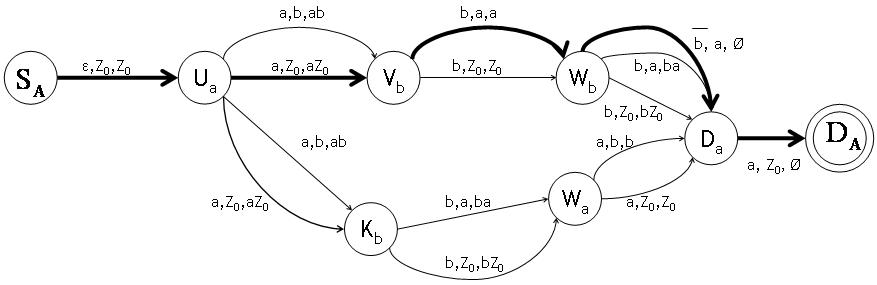}}
\subfigure[Transformed PDA]{\label{fig:tpda}\includegraphics[width=\textwidth]{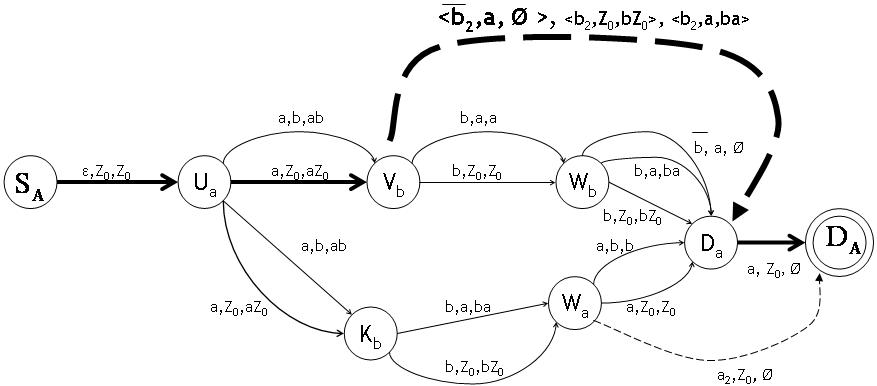}}
\subfigure[Subset of $G_N$ which generates $T_C$]{\label{fig:cfg}\includegraphics[scale=0.45]{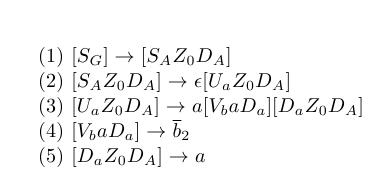}}
\caption{Example}
\end{figure}

%% file: extended_automata.bbl
\begin{thebibliography}{10}

\bibitem{RFC5659}
M.~Bocci and S.~Bryant.
\newblock {RFC5659 - An Architecture for Multi-Segment Pseudowire Emulation
  Edge-to-Edge}, 2009.

\bibitem{RFC3985}
S.~Bryant and P.~Pate.
\newblock {RFC3985 - Pseudo Wire Emulation Edge-to-Edge (PWE3) Architecture},
  2005.

\bibitem{Chlamtac1996}
I.~Chlamtac, A.~Farag{\'o}, and T.~Zhang.
\newblock {Lightpath (Wavelength) Routing in Large WDM Networks}.
\newblock {\em IEEE Journal on Selected Areas in Communications},
  14(5):909--913, 1996.

\bibitem{CHOKING}
H.~Cho, J.~Ryoo, and D.~King.
\newblock Stitching dynamically and statically provisioned segments to
  construct end-to-end multi-segment pseudowires.
\newblock
  \url{http://www.ietf.org/id/draft-cho-pwe3-mpls-tp-mixed-ms-pw-setup-01.txt},
  2011.

\bibitem{Denise2010}
A.~Denise, Y.~Ponty, and M.~Termier.
\newblock Controlled non-uniform random generation of decomposable structures.
\newblock {\em Theoretical Computer Science}, 411(40-42):3527--3552, 2010.

\bibitem{Dijkstra2008}
F.~Dijkstra, B.~Andree, K.~Koymans, J.~van~der Ham, P.~Grosso, and C.~de~Laat.
\newblock A multi-layer network model based on {ITU-T} {G}.805.
\newblock {\em Comput. Netw.}, 2008.

\bibitem{Dijkstra2009}
F.~Dijkstra, J.~Van der Ham, P.~Grosso, and C.~de~Laat.
\newblock A path finding implementation for multi-layer networks.
\newblock {\em Future Generation Comp. Syst.}, 25(2):142--146, 2009.

\bibitem{PCE}
A.~Farrel, JP. Vasseur, and J.~Ash.
\newblock {RFC4655 - A Path Computation Element (PCE)-Based Architecture},
  2006.

\bibitem{Flajolet94}
Ph. Flajolet, P.~Zimmermann, and B.~{Van Cutsem}.
\newblock A calculus for the random generation of labelled combinatorial
  structures.
\newblock {\em Theoretical Computer Science}, 1994.

\bibitem{Gardy2010}
D.~Gardy and Y.~Ponty.
\newblock {Weighted random generation of context-free languages: Analysis of
  collisions in random urn occupancy models}.
\newblock {\em {GASCom}}, 2010.

\bibitem{Gong2008}
S.~Gong and B.~Jabbari.
\newblock {Optimal and Efficient End-to-End Path Computation in Multi-Layer
  Networks}.
\newblock In {\em ICC}, pages 5767--5771, 2008.

\bibitem{Hickey83}
T.~Hickey and J.~Cohen.
\newblock Uniform random generation of strings in a context-free language.
\newblock {\em SIAM J. Comput.}, 12(4):645--655, 1983.

\bibitem{Hop06}
J.~E. Hopcroft, R.~Motwani, and J.~D. Ullman.
\newblock {From PDA's to Grammars}.
\newblock In {\em Introduction to Automata Theory, Languages, and Computation},
  chapter 6.3.2, pages 247--251. Addison-Wesley Longman Publishing Co., Inc.,
  Boston, MA, USA, 2006.

\bibitem{K09}
F.~A. Kuipers and F.~Dijkstra.
\newblock Path selection in multi-layer networks.
\newblock {\em Computer Communications}, 2009.

\bibitem{lamali2012}
M.~L. Lamali, H.~Pouyllau, and D.~Barth.
\newblock Path computation in {Multi-Layer} multi-domain networks.
\newblock In {\em IFIP/TC6 Networking 2012 (NETWORKING 2012)}, Prague, Czech
  Republic, May 2012.

\bibitem{LP11}
M.L. Lamali, H.~Pouyllau, and D.~Barth.
\newblock {End-to-End Quality of Service in Pseudo-Wire Networks}.
\newblock In {\em ACM CoNEXT Student Workshop}, 2011.

\bibitem{dyckwords}
Jens Liebehenschel.
\newblock {Lexicographical Generation of a Generalized Dyck Language}.
\newblock {\em SIAM J. Comput.}, 2003.

\bibitem{Mairson94}
H.~G. Mairson.
\newblock {Generating Words in a Context-Free Language Uniformly at Random}.
\newblock {\em Inf. Process. Lett.}, 49(2):95--99, 1994.

\bibitem{RFC4448}
L.~Martini, E.~Rosen, N.~El-Aawar, and G.~Heron.
\newblock {RFC4448} - {Encapsulation Methods for Transport of Ethernet over
  MPLS Networks}, 2008.

\bibitem{RFC5212}
K.~Shiomoto, D.~Papadimitriou, JL.~Le Roux, M.~Vigoureux, and D.~Brungard.
\newblock {RFC}5212 - {R}equirements for {GMPLS}-based multi-region and
  multi-layer networks ({MRN/MLN}), 2008.

\bibitem{RFC5087}
Y(J). Stein, R.~Shashoua, R.~Insler, and M.~Anavi.
\newblock {RFC}5087 - {Time Division Multiplexing over IP (TDMoIP)}, 2007.

\bibitem{Yao2005}
W.~Yao and B.~Ramamurthy.
\newblock {A Link Bundled Auxiliary Graph Model for Constrained Dynamic Traffic
  Grooming in WDM Mesh Networks}.
\newblock {\em IEEE Journal on Selected Areas in Communications},
  23(8):1542--1555, 2005.

\bibitem{Zhu2003}
H.~Zhu, H.~Zang, K.~Zhu, and B.~Mukherjee.
\newblock A novel generic graph model for traffic grooming in heterogeneous
  {WDM} mesh networks.
\newblock {\em IEEE/ACM Trans. Netw.}, 11(2):285--299, 2003.

\end{thebibliography}
